\documentclass[preprint]{elsarticle}

\usepackage{array,xspace,multirow,hhline,tikz,colortbl,tabularx,booktabs,fixltx2e,amsmath,amssymb,amsfonts,amsthm}
\usetikzlibrary{arrows}
\usepackage{algorithm}
\usepackage{algorithmic}
\usepackage{eqparbox}
\usepackage{verbatim,ifthen}
\usepackage{enumitem}
\usepackage{pifont}
\usepackage{ifthen}
\usepackage{calrsfs,mathrsfs}
\usepackage{bbding,pifont}
\usepackage{pgflibraryshapes}

	\usepackage{varioref}
		\usepackage{subfigure}

\definecolor{light-gray}{gray}{0.9}

	\newtheorem{theorem}{Theorem}%
	\newtheorem{corollary}{Corollary}%
	\newtheorem{example}{Example}

	\newcommand\eat[1]{}

	\usepackage{enumitem}
	\setenumerate[1]{label=\rm(\it{\roman{*}}\rm),ref=({\it\roman{*}}),leftmargin=*}
	\newlength{\wordlength}

	\newcommand{\set}[1]{\{#1\}}
	\newcommand{\midd}{\mathbin{:}}

	\newcommand{\eqclass}[2][]{\ifthenelse{\equal{#1}{}}{[#2]}{[#2]_{\sim_{#1}}}}

	\newcommand{\pref}{\succsim\xspace}
		\newcommand{\spref}{\succ\xspace}

\usepackage{enumitem}
\setenumerate[1]{label=\rm(\it{\roman{*}}\rm),ref=({\it\roman{*}}),leftmargin=*}

\newcommand{\nbh}[1][]{
	\ifthenelse{\equal{#1}{}}{\nu}{\nu(#1)}
}

\newcommand{\cstr}[1][]{
	\ifthenelse{\equal{#1}{}}{\mathscr S}{\cstr(#1)}
}

\newcommand{\choice}[1][]{
	\ifthenelse{\equal{#1}{}}{\mathit{C}}{\choice(#1)}

		\newcommand{\ml}[1][]{\ensuremath{\ifthenelse{\equal{#1}{}}{\mathit{ML}}{\mathit{ML}(#1)}}\xspace}
		\newcommand{\sml}[1][]{\ensuremath{\ifthenelse{\equal{#1}{}}{\mathit{SML}}{\mathit{SML}(#1)}}\xspace}
		\newcommand{\sd}[1][]{\ensuremath{\ifthenelse{\equal{#1}{}}{\mathit{SD}}{\mathit{SD}(#1)}}\xspace}
		\newcommand{\rsd}[1][]{\ensuremath{\ifthenelse{\equal{#1}{}}{\mathit{RSD}}{\mathit{RSD}(#1)}}\xspace}
		\newcommand{\rd}[1][]{\ensuremath{\ifthenelse{\equal{#1}{}}{\mathit{RD}}{\mathit{RD}(#1)}}\xspace}
		\newcommand{\st}[1][]{\ensuremath{\ifthenelse{\equal{#1}{}}{\mathit{ST}}{\mathit{ST}(#1)}}\xspace}
		\newcommand{\bd}[1][]{\ensuremath{\ifthenelse{\equal{#1}{}}{\mathit{BD}}{\mathit{BD}(#1)}}\xspace}
		\newcommand{\pc}[1][]{\ensuremath{\ifthenelse{\equal{#1}{}}{\mathit{PC}}{\mathit{PC}(#1)}}\xspace}
		\newcommand{\dl}[1][]{\ensuremath{\ifthenelse{\equal{#1}{}}{\mathit{DL}}{\mathit{DL}(#1)}}\xspace}
		\newcommand{\ul}[1][]{\ensuremath{\ifthenelse{\equal{#1}{}}{\mathit{UL}}{\mathit{UL}(#1)}}\xspace}

			\newcommand{\indiff}{\ensuremath{\sim}}}

\begin{document}

		\title{An Impossibility Result for Housing Markets with Fractional Endowments}



	\author{Haris Aziz\corref{cor1}} \ead{haris.aziz@unsw.edu.au}
		\address{UNSW Sydney and Data61 (CSIRO), Australia}



	\begin{keyword}
	 House allocation\sep 
	Housing markets\sep
	Core\sep
	Top Trading Cycles\sep 
		
		\emph{JEL}: C62, C63, and C78
	\end{keyword}

	\begin{abstract}
The housing market setting constitutes a fundamental model of exchange economies of goods. Most of the work concerning housing markets does not cater for randomized assignments or allocation of time-shares. House allocation with fractional endowments of houses was considered by Athanassoglou and Sethuraman (2011) who posed the open problem whether individual rationality, weak strategyproofness, and efficiency are compatible for the setting. We show that the three axioms are incompatible. 
	\end{abstract}

	\maketitle


		\sloppy

		\section{Introduction}


	
The housing market is a fundamental model of exchange economies of goods. It has been used to model online barter markets and nation-wide kidney markets~\citep{RSU07a,SoUn10a}. 
		The housing market (also called the \emph{Shapley-Scarf} market) 
		consists of a set of agents each of whom owns a house and has preferences over the set of houses. The goal is to redistribute the houses among the agents in an efficient and stable manner.
The desirable properties include the following ones: \emph{Pareto optimality} (there exists no other assignment which each agent weakly prefers and at least one agent strictly prefers); \emph{individual rationality} (the resultant allocation is at least as preferred by each agent as his endowment); and \emph{core stability} (there exists no subset of agents who could have redistributed their endowments among themselves so as to get a more preferred outcome than the resultant assignment). 		
		
\citet{ShSc74a} showed that for housing markets with strict preferences, an elegant mechanism called \emph{Gale's Top Trading Cycles (TTC)} (that is based on multi-way exchanges of houses between agents) is strategyproof and finds an allocation that is in the core~\citep{ShSc74a,Ma94a}.\footnote{The seminal paper of \citet{ShSc74a} was referenced prominently in the scientific background document of the \emph{2012 Sveriges Riksbank Prize in Economic Sciences in Memory of Alfred Nobel} given to Lloyd Shapley and Alvin Roth.} Along with the Deferred Acceptance Algorithm, TTC has provided the foundations for many of the developments in matching market design~\citep{Manl13a,SoUn10a}. The Shapley-Scarf market has been used to model important real-world problems for allocation of human organs and seats at public schools.
Since the formalization of TTC, considerable work has been done to extend and generalize TTC for more general domains that allow indifference in preferences~\citep{AlMo11a,JaMa12a,AzKe12a,Plax13a,SeSa13a} or multiple units in endowment~\citep{FLS+15a,KQW01a,STSY14a,TSY14a}.

		%
		%
		%
		%
		%
		%
		%

		Despite recent progress on house allocation and housing market mechanisms, the general assumption has remained that agents cannot own or be allocated fractions of houses. The disadvantage of this assumption is that it does not model various cases where agents have fractional endowments or when agents can share houses.  This is especially the case when agents have the right to use different facilities for different fractions of the time and fractional allocation of resources is helpful in obtaining more equitable outcomes. Allowing for exchanges of fractions of services can also be helpful in modeling time-bank models in which agents performs services in order to receive other services for the same time duration. 
Fractional allocation of houses can also be interpreted as the relative right of an agent over a house~\citep{AtSe11a}. 
Finally, fractional allocations can be used to model randomized allocation of indivisible resources where agents exchange probabilities of getting particular houses.
Hence allocation of houses under fractional endowments generalizes a number of well-studied house allocation models.

\citet{AtSe11a} considered housing markets with fractional endowments and presented a desirable mechanism for the setting. They posed the following open question regarding housing markets with fractional endowments: 
\emph{``a natural question to ask is whether there exists a mechanism that is individually rational, ordinally efficient, and weakly strategyproof.''}
In this paper, we answer the open question posed by \citet{AtSe11a} in the negative by presenting an impossibility result. The result implies a number of results in the literature including theorems in \citep{AtSe11a,Yilm10a}. The result also shows that as soon the Shapley-Scarf housing market allows for fractional endowments, the compatibility of individual rationality, Pareto optimality, and strategyproofness disappears.

		\section{Preliminaries}

		\subsection{Model}

		Consider a market with a set of agents $N=\{1,\ldots, n\}$ and a set of houses $H=\{h_1,\ldots, h_m\}$. Each agent has complete and transitive preferences $\pref_i$ over the houses and $\pref=(\pref_1,\ldots, \pref_n)$ is the preference profile of the agents.  
		Agents may be indifferent among houses.
			We denote $\pref_i: E_i^1,\dots,E_i^{k_i}$ for each agent $i$ with equivalence classes in decreasing order of preference.
			Thus, each set $E_i^j$ is a maximal equivalence class of houses among which agent $i$ is indifferent, and $k_i$ is the number of equivalence classes of agent $i$.
		An agent has \emph{dichotomous preferences} if he considers each house as either acceptable or unacceptable and is completely indifferent between unacceptable houses and also indifferent between acceptable houses. 

		Each agent $i$ is endowed with an allocation $e(i)$ where there are
		$e(i)(h_j)$ units of house $h_j$ given to agent $i$. 
		The quadruple $(N,H,\pref,e)$ is an instance of a  \emph{housing market with fractional endowments}. 
		Note that in the basic housing market, each agent is endowed with and is allocated one house and the endowments are discrete: $n=m$, $e(i)(h_j)\in \{0,1\}$ and $\sum_{h\in H}e(i)(h)=1$ for all $i\in N$ and $\sum_{i\in N}e(i)(h)=1$ for all $h\in H$. 
		When  allocations are discrete we will also abuse notation and denote $e(i)$ as a set.

		A \emph{fractional assignment} is an $n\times m$ matrix $[x(i)(h_j)]_{\substack{1\leq i\leq n, 1\leq j\leq m}}$ such that for all $i\in N$, and $h\in H$, $\sum_{i\in N}x(i)(h)=\sum_{i\in N}e(i)(h)$. 
		The value $x(i)(h_j)$ is the fraction or units of house $h_j$ that agent $i$ gets. We will use fraction or unit interchangeably since we do not assume that exactly one unit of each house is in the market. 
		Each row $x(i)=(x(i)(h_1),\ldots, x(i)(h_m))$ represents the \emph{allocation} of agent $i$.
		Given two allocations $x(i)$ and $x(j)$, $x(i)+x(j)$ is the point-wise sum of the allocations $x(i)$ and $x(j)$. If $\sum_{i\in N}x(i)(h)=1$ for each $h\in H$, a fractional assignment can also be interpreted as a random assignment where $x(i)(h_j)$ is the probability of agent $i$ getting house $h_j$.
		Note that endowment $e$ itself can be considered as the initial assignment of houses to the agents with $e(i)$ being the initial allocation of agent $i\in N$.
		A \emph{fractional housing market} mechanism is a function that takes as input $(N,H,\pref,e)$ and returns an assignment or vector of allocations $(x(1),\ldots, x(n))$ such that $\sum_{i\in N}x(i)=\sum_{i\in N}e(i)$.
		We do not require in general that $\sum_{i\in N}e(i)(h)$ or $\sum_{h\in H}e(i)(h)$  are integers.

%
%
%
%
%
%
%
%
%
%
		
				\begin{example}[Illustration of a housing market with fractional endowments]
				Consider the following housing market $(N,H,\succsim ,e)$ where $N=\{1,2,3\}$ and $H=\{a,b,c\}$. 
				The endowment assignment and the preferences are as follows:

						\[
						e=\begin{pmatrix}
						0&0.99&0.01\\
						0.99&0&0.01\\
						 0.01&0.01&0.98\\
						   		 \end{pmatrix}.
						\]
		
						\begin{align*}
							1:&\quad a,c,b&
							2:&\quad b,a,c&
							3:&\quad b,a,c
							\end{align*}
							The endowment assignment specifies that agent $1$ owns 0.99 fraction of $b$ and 0.01 fraction of $c$. 
							
				\end{example}

		%
		%
		%
		%
		%
		%
		%
		%

		%
		%
		%

		\subsection{Properties of allocations and mechanisms}


		Before defining various stability and efficiency properties, we need to reason about agents' preferences over allocations. A standard method to compare random allocations is to use the \emph{SD (stochastic dominance)} relation. 
					%
					 Given two random assignments $x$ and $y$, $x(i) \succsim_i^{SD} y(i)$ i.e.,  an agent  $i$ \emph{SD~prefers} allocation $x(i)$ to $y(i)$ if
\[\forall h\in H: \textstyle \sum_{h_j\in \set{h_k\midd h_k\pref_i h}}x(i)(h_j) \ge \sum_{h_j\in \set{h_k\midd h_k\pref_i h}}y(i)(h_j).\] 


					
					The SD relation is not complete in general. 
We define normative properties of allocations as well as mechanisms.

		\begin{itemize}
		\item \emph{$SD$-efficiency}: an assignment $x$ is\emph{ $SD$-efficient} if there exists no other assignment $y$ such that $y(i)\succsim_i^{SD} x(i)$ for all $i\in N$ and $y(i)\succ_i^{SD} x(i)$ for some $i\in N$.
		\item \emph{$SD$ individually rational (SD-IR)}: an assignment $x$ is \emph{$SD$-individually rational} if $x(i)\succsim_i^{SD} e(i)$.
		\item A mechanism~$f$ is ${SD}$-\emph{manipulable} iff there exists an agent $i \in N$ and preference profiles $\pref$ and $\pref'$ with $\pref_j=\pref_j'$ for all $j \neq i$ such that $f(\pref')$ 
		$\spref_i^{SD} f(\pref)$. 
		A mechanism is \emph{weakly} ${SD}$-\emph{strategyproof} iff it is not ${SD}$-manipulable.
		\end{itemize}

		\section{The result}

%
%
%
%
%
%
%
%
%

We show that there does not exist an SD-efficient, SD-individually rational and weak SD-strategyproof mechanism. The theorem below answers a question raised by \citet{AtSe11a}.

		\begin{theorem}\label{th:imposs}
			There does not exist a weak SD-strategyproof, SD-efficient and SD-individually rational fractional housing market mechanism even for single unit allocations and endowments, and for strict preferences. 
			\end{theorem}

			\begin{proof}
				Consider the housing market $(N,H,\pref,e)$ where $N=\{1,2,3,4,5\}$, $H=\{h_1,h_2,h_3, h_4,h_5\}$, the preference profile $\pref$ is
				\begin{align*}
					\pref_1: &\quad h_3, h_1, h_2,h_4,h_5\\
					\pref_2: &\quad h_5,h_1,h_2,h_3,h_4\\
					\pref_3: &\quad h_1,h_4,h_2,h_3,h_5\\
					\pref_4: &\quad h_2,h_4,h_1,h_3,h_5\\
						\pref_5: &\quad h_5,h_3,h_1,h_2,h_4
					\end{align*}

				and

					\[
					e=\begin{pmatrix}
					1/2&1/2&0&0&0\\
					0&0&1/2&0&1/2\\
						1/2&0&0&1/2&0\\
						0&1/2&0&1/2&0\\
							0&0&1/2&0&1/2\\
					   		 \end{pmatrix}.
					\]

		First we can establish the claim that agents $3, 4, 5$ get exactly their endowment as long as the following conditions hold:

		\begin{itemize}
			\item[A] the allocation is SD-IR.
			\item[B] agents $3,4,$ and $5$ report truthfully.
			\item[C] agent $2$ reports $h_4$ as his least preferred house.
			\item[D] agent $2$ reports $h_5$ as his most preferred house.
			\item[E] agent $1$ expresses $h_4$ as a house less preferred than $h_1,h_2,h_3$.
		\end{itemize}

	We argue for the claim as follows. 
	Suppose the allocation is 

						\[
						a=\begin{pmatrix}
						a_{11}&a_{12}&a_{13}&a_{14}&a_{15}\\
						a_{21}&a_{22}&a_{23}&a_{24}&a_{25}\\
						a_{31}&a_{32}&a_{33}&a_{34}&a_{35}\\
						a_{41}&a_{42}&a_{43}&a_{44}&a_{45}\\
							a_{51}&a_{52}&a_{53}&a_{54}&a_{55}\\
						   		 \end{pmatrix}.
						\]
	Since $h_5$ is the most preferred house of agent 2 (Assumption D) and agent 5 (truthful preference of agent 5 according to Assumption B) who also happen to hold 0.5 each of $h_5$, SD-IR of allocation $a$ implies that $a_{25}=a_{55}=0.5$ and $a_{15}=a_{35}=a_{45}=0$. House $h_4$ is the second most preferred house of agents 3 and 4 and they get 0.5 of their second most preferred house and 0.5 of their most preferred house. If agent 3 or agent 4 get an SD-improvement over their endowment, then they must get less of $h_4$ and more of their most preferred house. However, due to B, E, and the assumption that agent $5$ reports truthfully,  no agent among 1, 2, and 5 can take any fraction of $h_4$ or else SD-IR is violated. Hence,  $a_{34}=a_{44}=0.5$ and  $a_{14}=a_{24}=a_{54}=0$.
Since $h_1$ is the most preferred house of agent $3$ and her allocation of $h_4$ is fixed, it follows from SD-IR that $a_{31}=0.5$. Hence $a(3)=e(3)$. Since $h_2$ is the most preferred house of agent $4$ and her allocation of $h_4$ is fixed, it follows from SD-IR that $a_{42}=0.5$. Hence $a(4)=e(4)$. It remains to be established that $a_{53}=0.5$. Since $a_{55}=0.5$ is already fixed and since $h_3$ is the next most preferred house of agent 5 after house $h_5$, SD-IR implies that $a_{53}=0.5$. Thus $a(5)=e(5)$ and we have proved the claim.

		From now on, we will consider a preference profile in which the conditions above are met so that by SD-IR, we get that agents 3, 4, 5 get exactly their endowments.


		Assuming that agents $3,4,5$ get the same allocation as their endowment, agent $1$ must get $1/2$ of $h_3$ in any SD-efficient assignment. Thus the only SD-individually rational and SD-efficient assignments for profiles satisfying the conditions above:

					\[
					x=\begin{pmatrix}
					1/2&0&1/2&0&0\\
					0&1/2&0&0&1/2\\
					1/2&0&0&1/2&0\\
					0&1/2&0&1/2&0\\
						0&0&1/2&0&1/2\\
					   		 \end{pmatrix},
					\]

and
					\[
					z=\begin{pmatrix}
					\lambda&1/2-\lambda&1/2&0&0\\
					1/2-\lambda&\lambda&0&0&1/2\\
					1/2&0&0&1/2&0\\
					0&1/2&0&1/2&0\\
						0&0&1/2&0&1/2\\
					   		 \end{pmatrix}
					\]

		for $0\leq \lambda < 1/2.$

		If the outcome is assignment $x$, then agent $2$ can report $\pref_2'$:
			\begin{align*}
				\pref_1:&\quad h_3, h_1, h_2,\ldots\\
				\pref_2':&\quad h_5,h_1,h_3,h_2,h_4
				\end{align*}
				The only SD-individually rational and SD-efficient outcome of $(\pref_1,\pref_2',\pref_3,\pref_4,\pref_5)$ is assignment $y$ which is an SD-improvement for agent $2$ over the truthful outcome $x$.


								If the outcome is of type assignment $z$, then  						agent $1$ can report $\pref_1'$:
\begin{align*}
											\pref_1':&\quad h_1, h_3, h_2,\ldots\\
											\pref_2:&\quad h_5,h_1,h_2,h_3,h_4
											\end{align*}
											The only SD-individually rational and SD-efficient outcome of $(\pref_1',\pref_2,\pref_3,\pref_4,\pref_5)$ is assignment $x$ which is an SD-improvement for agent $1$ over the truthful outcome $z$.
 						\end{proof}

					We also get as corollaries previous impossibility results in the literature:

					\begin{corollary}[Theorem 4, \citet{Yilm10a}]
						There does not exist an SD-IR, SD-efficient, weak SD-strategyproof, and weak SD-envy-free fractional housing market mechanism.
						\end{corollary}
				
						\begin{corollary}[Theorem 2, \citet{AtSe11a}]
							There does not exist an SD-IR, SD-efficient, weak justified envy-free and weak SD-strategyproof fractional housing market mechanism.
							\end{corollary}
				
						\begin{corollary}[Theorem 3, \citet{AtSe11a}]
							There does not exist an SD-IR, SD-efficient, and SD-strategyproof fractional housing market mechanism.
							\end{corollary}
				
	We remark that the three properties used in Theorem~\ref{th:imposs} are independent from each other. 
SD-efficiency and weak SD-strategypoofness can be simultaneously satisfied by the \emph{multi-unit eating probabilistic serial mechanism}~\citep{Koji09a} if preferences are strict.
SD-individual rationality and weak SD-strategyproofness (even SD-strategyproofness) are satisfied by the mechanism that returns the endowment. 
SD-individual rationality and SD-efficiency can be satisfied by imposing linear constraints for SD-IR and then maximizing sum of utilities that are consistent with the ordinal preferences.



	%


%

		\end{document}